\documentclass{article}

\usepackage{float}
\usepackage{amsmath}
\usepackage{graphicx}
\usepackage{float}
\usepackage{geometry}
\usepackage{enumerate}
\usepackage[english]{babel}
\usepackage{longtable}
\usepackage{color}
\usepackage{amssymb, amsmath, amsthm}
\usepackage{multirow}
\usepackage[titletoc]{appendix}
\usepackage[numbers]{natbib}
\RequirePackage[OT1]{fontenc}
\usepackage{amsthm,amsmath}
\usepackage{setspace}
\usepackage{authblk}

\RequirePackage[colorlinks,citecolor=blue,urlcolor=blue]{hyperref}

\newtheorem{theorem}{Theorem}

\DeclareMathOperator{\expit}{expit}
\DeclareMathOperator{\logit}{logit}

\newcommand{\pt}{\mbox{$p_0$}}
\newcommand{\Pt}{\mbox{$P_0$}}

\newcommand{\indep}{\mbox{$\perp\!\!\!\perp$}}

\renewenvironment{proof}{{\it Proof }}{\qed \\}

\title{Targeted Maximum Likelihood
  Estimation using Exponential Families}
\author[1]{Iv\'an  D\'iaz\thanks{idiaz@jhu.edu}}
\author[1]{Michael Rosenblum\thanks{mrosenbl@jhsph.edu}}
\affil[1]{\small Department of Biostatistics, Johns Hopkins Bloomberg School
  of Public Health.}

\begin{document}\maketitle
\onehalfspacing

\begin{abstract}
  Targeted maximum likelihood estimation (TMLE) is a general method for
  estimating parameters in semiparametric and nonparametric
  models. Each iteration of TMLE involves fitting a parametric
  submodel that targets the parameter of interest.
  We investigate the use of exponential families to
  define the parametric submodel. This implementation of TMLE gives a general
  approach for estimating any smooth parameter in the nonparametric model.
  A computational advantage of this approach is that each iteration of TMLE involves estimation of a parameter
  in an exponential family, which is a convex optimization problem for which software implementing reliable and computationally efficient methods
  exists. We illustrate the method in three estimation problems, involving
  the mean of an outcome missing at random, the parameter of a median
  regression model, and the causal effect of a continuous
  exposure, respectively.
 We conduct a simulation study comparing different choices for the parametric submodel, focusing on the first of these problems.
To the best of our knowledge, this is the first study investigating robustness of TMLE to
different specifications of the parametric submodel. We find that the
choice of submodel can have an important impact on the behavior of the
estimator in finite samples.\\
Key words: TMLE, Exponential family, Convex optimization.
\end{abstract}

\section{Introduction}\label{intro}

Targeted maximum likelihood estimation \cite[TMLE,][]{vanderLaan&Rubin06,
  vanderLaanRose11} is a general template for
estimation in semiparametric or nonparametric models.
The key to each update step of TMLE is to specify and fit a parametric submodel with certain properties, described in detail below.
(A parametric submodel is defined as a parametric model that is contained in the overall model.) Throughout this paper, the overall model is a nonparametric model.
Therefore, many choices are available for the form of the parametric submodel.

We investigate the performance of TMLE
when the parametric submodel is chosen to be from an exponential
family. A computational advantage of this choice is that because the parametrization as well
as the parameter space of an exponential family are always convex
\citep{boyd2004}, standard methods for optimization can be applied to solve this
problem. Another advantage of this approach is that it can be applied to a wide variety of estimation problems.
Specifically, it can be used to estimate any smooth parameter defined in the nonparametric model, under conditions described below.

We demonstrate how to implement TMLE using exponential families
in the following three estimation problems:
\begin{enumerate}
\item Estimating the mean of an outcome missing at random, where covariates are observed for the entire sample.
\item Estimating a nonparametric extension of the parameter in a median regression model. 
\item Estimating the causal effect of a continuous-valued exposure.
\end{enumerate}

We conduct a simulation study comparing different choices for the parametric submodel, focusing on the first of these problems.

We next present an overview of the TMLE template, and illustrate the implementation of TMLE for each of the
three problems above.  We conclude with a
discussion of practical issues and directions for future research.

\section{Targeted maximum likelihood template}\label{tmle}
Let the random vector representing what is observed on an experimental unit be denoted by $O$, with sample space $\mathcal{O}$.
Let $\{O_1,\ldots,O_n\}$ be an independent, identically distributed sample
of observations $O$, each drawn from the unknown, true distribution $\Pt$. We assume that $\Pt\in \mathcal M$, where
$\mathcal M$ is the nonparametric model, defined as the class of all distributions having a continuous density with respect to
a dominating measure $\nu$.
Let $\mathcal{M}'$ denote the class of all densities corresponding to a distribution in $\mathcal{M}$, and let $p_0$ denote the density corresponding to $\Pt$.
Let $\Psi(p)$ denote a $d$-dimensional Euclidean parameter with known
efficient influence function $D(\pt,O)$. That is, $\Psi$ is a mapping from $\mathcal{M}$ to $\mathbb{R}^d$ for which
$D(\pt,O)$ is the pathwise derivative, as defined, e.g., in \citet{Bickel97}. We refer to such a parameter as a smooth parameter.
Many commonly used parameters are smooth, including all those in this paper. The efficient influence function, by definition, satisfies
\begin{equation}
  \int D(p,o)p(o)d\nu(o)=0 \mbox{ for all } p \in\mathcal{M}'. \label{effic_property}
\end{equation}


Denote the true value $\Psi(\pt)$ by $\psi_0$. The template for a
targeted maximum likelihood estimator is defined by the following steps:
\begin{enumerate}
\item Construct an initial estimator $p^0$
  of the true, unknown density $\pt$;
\item Construct a sequence of updated density estimates $p^k$, $k=1,2,\dots$. Given the current density estimate $p^k$, the updated density $p^{k+1}$ is constructed by specifying
  a regular parametric submodel $\{p_\epsilon^k:\epsilon \in R\}$ of $\mathcal{M}$, where  $R$  is an open subset of $\mathbb{R}^d$.
  The submodel is required to satisfy two properties. First, it must  equal the current density estimate $p^k$ at
  $\epsilon=0$, i.e., $p_0^k = p^k$. Second, the score of $p_\epsilon^k$ at $\epsilon=0$ must equal the efficient
  influence function for  $\Psi$ at $p^k$, i.e.,
  \begin{equation}
    D(p^k,o) = \frac{d}{d
      \epsilon}\left[\log p_\epsilon^k(o)\right]\bigg|_{\epsilon=0}, \mbox{ for all possible values of } o \in \mathcal{O}.\label{score}
  \end{equation}
  The parameter $\epsilon$  of the submodel $\{p_\epsilon^k:\epsilon \in R\}$ is fit using maximum likelihood estimation, i.e.,
  \begin{equation}
    \hat\epsilon=\arg\max_\epsilon\sum_{i=1}^n\log p^k_\epsilon(O_i), \label{mle}
  \end{equation}
  and
  the updated density $p^{k+1}$ is defined to be the density in the parametric submodel\\ $\{p_\epsilon^k:\epsilon \in \mathbb{R}^d\}$ corresponding to $\hat\epsilon$, i.e.,
  $p^{k+1}=p^k_{\hat\epsilon}$.
\item Iterate the previous step  until convergence, i.e., until $\hat\epsilon\approx 0$. Denote the last step of the procedure by $k=k^*$.
\item Define the TMLE of $\psi_0$ to be
the substitution estimator
  $\hat\psi\equiv \Psi(p^{k^*})$.
\end{enumerate}
The TMLE algorithm above can be generalized in the following ways: one
can let each $p^k$ represent an estimate of only certain components
of the density $p$ (typically those components relevant to estimation
of the parameter $\Psi$ for a specific problem); the parametric
submodel may satisfy a relaxed score condition, in that  the efficient
influence function need only be contained in the linear span of the
score at $\epsilon=0$; or another loss function may be used in place
of the log-likelihood for estimating $\epsilon$ in (\ref{mle}). We
discuss the latter generalization in more detail in Section~\ref{second_onestep}.

The result of the above TMLE procedure is that at the final density estimate $p^{k^*}$, we have
\begin{equation}
  \sum_{i=1}^n D(p^{k^*},O_i) \approx 0, \label{eifee}
\end{equation}
i.e., the final density estimate is a solution to the efficient influence function estimating equation.
This property, combined with the estimator being a substitution
estimator $\hat\psi\equiv \Psi(p^{k^*})$ and the smoothness of the parameter, is fundamental to proving
that TMLE has desirable properties. For example, it is
asymptotically linear with influence function equal
to the efficient influence function under certain assumptions, as described by  \cite{vanderLaan&Rubin06}.

We give a heuristic argument for (\ref{eifee}), which is rigorously
justified under regularity conditions given in Result 1 of \cite{vanderLaan&Rubin06}.
Assume that at the final iteration of TMLE, i.e., the iteration where
$p^{k^*}$ is defined, we have $\hat{\epsilon}=0$. Then by equation (\ref{mle}),
the penultimate density $p^{k^*-1}$ must satisfy
\begin{equation}
  0 = \hat\epsilon=\arg\max_\epsilon\sum_{i=1}^n\log p^{k^*-1}_\epsilon(O_i). \label{mle1}
\end{equation}
If  $\log p^{k^*-1}_\epsilon$ is strictly convex in $\epsilon$, then the derivative at $\epsilon=0$ of the right  side of (\ref{mle1}) equals $0$, which implies
\begin{eqnarray}
  0 & = &\sum_{i=1}^n  \frac{d}{d \epsilon} \log p^{k^*-1}_\epsilon(O_i)\bigg|_{\epsilon=\mathbf{0}} =   \sum_{i=1}^n  D(p^{k^*-1},O_i) = \sum_{i=1}^n  D(p^{k^*},O_i), \label{used_pass_through_at_epsilon_0_property}
\end{eqnarray}
where the second equality follows from the score condition (\ref{score}) and third equality follows since by construction in step 2 we have
$p^{k^*}=p^{k^*-1}_{\hat\epsilon} = p^{k^*-1}_{0} = p^{k^*-1}$. This completes the heuristic argument for (\ref{eifee}).

\section{Implementing TMLE using an exponential family}
The key step in the TMLE algorithm is step 2, which requires a choice of the parametric submodel at each iteration $k$.
Let $p^k$ denote the density at  the current iteration $k$, and consider construction of the parametric model in step 2.
One option is to use a submodel from an
exponential family, represented as
\begin{equation}  p_\epsilon^k(O)  = c(\epsilon, p^k)\exp\{\epsilon
  D(p^k,O)\}p^k(O),\label{expfam1}
\end{equation}
where the normalizing constant $c(\epsilon, p^k) = \left[ \int \exp\{\epsilon
D(p^k,o)\}p^k(o) d\nu(o)\right]^{-1}$. The model is defined for all $\epsilon \in \mathbb{R}^d$ for which the integral in $c(\epsilon, p^k)$ is finite.

A key feature of parametric models of the form (\ref{expfam1}) is that they automatically satisfy the two conditions in step 2 of the TMLE algorithm.
First, by (\ref{expfam1}), we have at $\epsilon=0$ that $p_\epsilon^k  =  p^k$. Second, the score condition (\ref{score}) holds
since we have, for all possible values of  $o \in \mathcal{O}$,
\begin{eqnarray}
  \frac{d}{d\epsilon} \left\{\log p_\epsilon^k (o)\right\}\bigg|_{\epsilon =0} & = &
  \frac{d}{d\epsilon} \left[ \log
    \exp\left\{\epsilon
      D(p^k,o)\right\} \right] \bigg|_{\epsilon =0}+ \frac{d}{d\epsilon} \left\{ \log c(\epsilon, p^k)\right\} \bigg|_{\epsilon =0} \nonumber \\
  & = &
  D(p^k,o) -
  \left[ c(\epsilon, p^k)\int D(p^k,o')\exp\{\epsilon D(p^k,o')\}p^k(o')d\nu(o')\right]\bigg|_{\epsilon =0} \nonumber \\
  &=&
  D(p^k,o) -
 c(0, p^k) \int D(p^k,o')p^k(o')d\nu(o')\nonumber \\
  &=&
  D(p^k,o), \nonumber
\end{eqnarray}
where the second equality follows from exchanging the order of differentiation and integration (justified under smoothness conditions by Fubini's Theorem), and the last equality follows from
$\int D(p^k,o)p^k(o)d\nu(o)=0$ by (\ref{effic_property}).

An advantage of using (\ref{expfam1}) as a submodel is that maximum likelihood
estimation of $\epsilon$ in an exponential family is a convex
optimization problem, which is computationally tractable. In particular, we take advantage of
the various R functions available to solve convex optimization
problems.

In Section~\ref{simula}, we illustrate a different implementation of TMLE that uses parametric
submodels given by
\begin{equation}
  p_\epsilon^k(O)=  c'(\epsilon, p^k)\{1 + \exp[-2\epsilon D(p^k,O)]\}^{-1}p^k(O), \label{ivan_submodel}
\end{equation}
for $c'(\epsilon, p^k)$ a normalizing constant, and which also has a convex log-likelihood function.

In Sections~\ref{ex1}-\ref{ex3}, we describe three estimation problems that can be solved with TMLE. The first problem,
estimating the mean of a variable missing at random, is an example where there exists a TMLE implementation that requires a single iteration and
that can be solved through a logistic regression of the outcome on a
so called ``clever covariate.'' In contrast, the examples in Sections~\ref{ex2}-\ref{ex3} generally require multiple iterations,
and cannot be solved using a clever covariate.

\section{Example 1: The mean of a variable  missing at
  random}\label{ex1}
  \subsection{Problem Definition}\label{sec_prob_defn_1}
Assume we observe $n$ independent, identically distributed draws $O_1,\ldots,O_n$, each
having the observed data structure $O=(X,M,MY)\sim\Pt$, where
$X$ is a vector of baseline random variables,
$M$ is an indicator of the outcome being observed, and $Y$ is the binary outcome.
For participants with $M=0$, we do not observe their outcome $Y$ (since we only observe $MY$, which equals zero for such participants); however, these participants do contribute baseline variables.
The only assumptions we make on the joint distribution
of $(X,M,Y)$ are that $Y$ is missing at random conditioned on $X$,
i.e., $M\indep Y|X$, and that $P(M=1|X)>0$ with probability 1.

Define the outcome regression $\mu(X) \equiv E_P(Y|M=1,X)$, the propensity score $p_M(X)\equiv P(M=1|X)$, and
the marginal density $p_X(X)$ of the baseline variables $X$. All of these components of the density $p$ are assumed unknown.
The parameter of interest is $E_{P}(Y)$, which by the missing at random assumption equals
$E_{p_X}(\mu(X))$. This parameter  only depends on the components $p_X$ and $\mu$ of the joint distribution $P$.
We  denote the parameter of interest as
$\Psi(\mu, p_X)=E_{P}(Y)\equiv E_{p_X}(\mu(X))$, where $E_{p_X}$ denotes the expectation with respect to the marginal distribution of $X$.

Note that in general
$E_{p_X}(\mu(X))  \neq  E_{P}(Y|M=1)$ (since the latter equals $E_{p_{X|M=1}}(\mu(X))$, i.e., the expectation of $\mu(X)$ with respect to the distribution of $X$ given $M=1$)
except in the special case
called missing completely at random, where $M$ and $X$ are marginally independent.
Denote the true mean of $Y$ by $\psi_0$. Identification and
estimation of $\psi_0$ is a widely studied problem
\cite[e.g.,][]{Robins97R,Kang&Shafer07}. The estimation problem
becomes particularly challenging when the dimension of $X$ is large,
since nonparametric estimators using empirical means suffer from the
curse of dimensionality. It is a challenging problem even when $X$ consists of a few, continuous-valued, baseline variables, as shown by \cite{Robins97R}.

Below, we contrast four TMLE implementations
for the above estimation problem. The purpose is to compare multiple options for the parametric submodel, in this simple problem. Also, we
demonstrate the general approach of using the exponential family
(\ref{expfam1}) as parametric submodel, in this relatively
well-studied problem, before applying it to more challenging problems
in Sections~\ref{ex2} and \ref{ex3}.
In Section~\ref{sec_onestep} we present an implementation of TMLE  from
\citet{vanderLaan&Rubin06}, which requires only a single iteration. A
variation of this estimator that uses weighted logistic
regression is presented in  Section~\ref{second_onestep}.
In Section \ref{ex1_exp_family_implementation}, we describe a TMLE implementation using the
exponential family (\ref{expfam1}) as a submodel; we also present a fourth implementation using the submodel (\ref{ff2}).

Under the conditions described in Theorem~\ref{Theo} of Appendix~\ref{theorem}, the
asymptotic distribution of the TMLE estimator
is not sensitive to the choice of submodel when each of the initial estimators $p_M^0$
and $\mu^0$ converges to its true value at a rate faster than
$n^{1/4}$. However, the choice of submodel may impact finite sample performance.
Also, when one of the estimators $p_M^0$ and $\mu^0$ does not converge to
the true value, the submodel choice may even affect performance asymptotically.
To shed light on this, we perform a simulation study in Section~\ref{simula}.

In general, TMLE implementations require that one has derived the
efficient influence function of the parameter of interest with respect
to the assumed model (which is the nonparametric model throughout this
paper). For the parameter in this section, the efficient influence
function is given by \cite{Bang05}
\begin{equation}
  D(p,O)=\frac{M}{p_M(X)}\{Y-\mu(X)\} + \mu(X) - \Psi(\mu, p_X).\label{eicEY}
\end{equation}
We will also
denote $D(p, O)$ by $D(\mu, p_M, p_{X},O)$, using the fact that the density $p$
can be decomposed into the components $\mu, p_M, p_{X}$ defined
above.

\subsection{First TMLE implementation for the mean of an outcome missing at random} \label{sec_onestep}

The first TMLE implementation has been extensively discussed in the literature
\cite[e.g.,][]{vanderLaan&Rubin06, vanderLaanRose11}, and we only provide
a brief recap. Following the template from Section~\ref{tmle},
we first define initial estimators $\mu^0$ and $p_M^0$ of
$\mu$ and $p_M$, respectively. These could obtained, e.g., by fitting logistic
regression models or by machine learning methods. We set the initial
estimator $p_{X}^0$  of $p_{X}$ to be the empirical distribution of
the baseline variables $X$, i.e., the distribution placing mass $1/n$
on each observation of $X$. The TMLE presented next is equivalent to
the estimator presented in page 1141 of \cite{Scharfstein1999R} when
$\mu^0$ is a logistic regression model fit. A  detailed discussion of the
similarities between the TMLE template and the estimators that stem
from \cite{Scharfstein1999R} is presented in Appendix 2 of
\cite{Rosenblum2010T}. We now show the construction of a parametric
submodel for this problem.
\paragraph{Construction of the parametric  submodel}
In this implementation of TMLE, the components $p_X$, $p_M$, and $\mu$
are each updated separately, such that they solve the corresponding part of the efficient influence function
estimating equation. Consider the $k$-th step of the TMLE algorithm
described in Section~\ref{tmle}. For the conditional expectation of $Y$ given
$X$ among individuals with $M=1$, and an estimator $\mu^k$, we define the logistic model
\begin{equation}
  \logit \mu_\epsilon^k(X) = \logit \mu^k(X) + \epsilon
  H_Y(X),\label{expFMAR}
\end{equation}
where $H_Y(X)=1/p_M^0(X)$. For the marginal distribution of $X$ we define the exponential model
\[p_{X,\theta}^k(X) \propto \exp\{\log p_{X}^k(X) + \theta H_X^k(X)\},\]
where $H_X^k(X)= \mu^k(X) - E_{p_X^k}\mu^k(X)$. The
variable $H_Y(X)$ has often been referred to as the ``clever
covariate''. The initial estimator of $p_M$ is not modified.

It is straightforward  to show that the efficient influence
function $D(p^k,O)$ is a linear combination of the scores of this joint
parametric model for the distribution of $O$.

We now describe the TMLE implementation based on the submodel construction above.
For this case, this procedure involves
only one iteration. In the first iteration we have
\[\hat \theta=\arg\max_{\theta}\frac{1}{n} \sum_{i=1}^n \log p_{X,\theta}^0(X_i)=0.\] This is because
the MLE of $p_{X}$ in the nonparametric model is precisely
the empirical $p_X^0$. An estimate  $\hat\epsilon$ of the parameter in the model
\[
  \logit \mu_\epsilon^0(X) = \logit \mu^0(X) + \epsilon
  H_Y(X),
\]
may obtained by running a logistic regression of $Y$ among individuals with $M=1$ on $H_Y(X)$ without intercept and including an
offset variable $\logit \mu^0(X)$. We now compute the updated estimate
of $\mu_0$ as \[\mu^1(X)=\expit \{\logit \mu^0(X) +
\hat\epsilon/p_M^0(X)\}.\] The  score equation corresponding to this logistic regression model is
\begin{equation}
\sum_{i=1}^n \frac{M_i}{p_M^0(X_i)}(Y_i - \mu^1(X_i))=0.\label{esteqTMLE1}
\end{equation}
Note that this matches the first component of the efficient influence function (\ref{eicEY}).
Proceeding to the second iteration, we
estimate the parameter $\epsilon$ in the model
\[\mu_\epsilon^1(X)=\expit \{\logit \mu^1(X) +
\epsilon /p_M^0(X)\},\]
by running a logistic regression of $Y$ on $H_Y(X)$ with offset $\mu^1(X)$
and without intercept among participants with $M=1$. This is equivalent to solving the score
equation
\[\sum_{i=1}^n \frac{M_i}{p_M^0(X_i)}(Y_i - \expit \{\logit \mu^1(X_i) +
\epsilon /p_M^0(X_i)\})=0\]
in $\epsilon$. By convexity and (\ref{esteqTMLE1}), the solution $\epsilon$ to the score equation in the above display
is equal to zero, and so the algorithm is terminates in the second iteration. The TMLE $\hat\psi$ is thus defined as
$\hat\psi=\frac{1}{n}\sum_{i=1}^n\mu^1(X_i)$. Confidence intervals may be constructed using
the non-parametric bootstrap. For a more detailed discussion of the
asymptotic properties of this estimator, as well as simulations, see \cite{vanderLaan&Rubin06}.

\subsection{Second TMLE implementation for the mean of an outcome missing at random} \label{second_onestep}
Consider the following $k$-th iteration parametric submodel for the
expectation of $Y$ conditional on $X$ among participants with $M=1$
\begin{equation}
  \logit \mu_\epsilon^k(X) = \logit \mu^k(X) + \epsilon\label{expFMARweights}.
\end{equation}
Let $\hat \epsilon$ be the first-step estimator of the intercept term
in a weighted logistic regression with weights $1/p_M^0(X)$ and offset
variable $\logit\mu^k(X)$. Let the updated estimator of $\mu_0$ be defined
as
\[\mu^1(X)=\expit \{\logit \mu^0(X) + \hat\epsilon\}.\]
By a similar argument as in the previous section, the estimate of
$\epsilon$ in the following iteration is equal to zero, and the TMLE
of $\psi$, defined as $\hat\psi=\frac{1}{n}\sum_{i=1}^n\mu^1(X_i)$,
converges in one step. Note that this implementation of the TMLE also
satisfies the score equation (\ref{esteqTMLE1}).

In addition, note that $\hat\epsilon$ in this section does not correspond to the
MLE of a parametric submodel. As a consequence, $\hat\psi$ is
not a targeted maximum likelihood estimator as defined in
Section~\ref{tmle}. Instead, it is part of a broader class of
estimators referred to as targeted minimum loss-based estimators, also
abbreviated as TMLE \cite{vanderLaanRose11}. These estimators generalize the TMLE framework of
Section~\ref{tmle} by allowing the use of general loss functions in
estimation of the parameter $\epsilon$ in the parametric submodel. In
the example of this section the loss function used is the weighted
least squares loss function.

This TMLE implementation is analogous to the estimator of Marshall
Joffe discussed in \cite{Robins2007} when $\mu^0$ is a parametric
model. The Joffe estimator is presented in \cite{Robins2007} as a
doubly robust alternative to the augmented IPW estimators when the
weights $1/p_M^0(X)$ are highly variable, i.e., when there are
empirical violations to the positivity
assumption; we simulate such scenarios in Section~\ref{simula}.
To the best of our knowledge, the above TMLE implementation was first
discussed by \citet{Stitelman2012} in the context of longitudinal
studies.

\subsection{Third and fourth TMLE implementations for the mean of an outcome missing at random} \label{ex1_exp_family_implementation}
We next give implementations of TMLE based on the following two types of submodel for $p$:
\begin{align}
  p_\epsilon^k(O)&=  c(\epsilon, p^k)\exp\{\epsilon D(p^k,O)\}p^k(O);\label{ff1}\\
  p_\epsilon^k(O)&=  c'(\epsilon, p^k)\{1 + \exp[-2\epsilon D(p^k,O)]\}^{-1}p^k(O).\label{ff2}
\end{align}
Here $c(\epsilon, p^k), c'(\epsilon, p^k)$ are the corresponding
normalizing constants, and $D(p^k,O)$ is the efficient influence
function given in (\ref{eicEY}). Model (\ref{ff1}) is the general
exponential family introduced in (\ref{expfam1}), while (\ref{ff2}) is
an alternative submodel.



The third and fourth TMLE implementations for estimating  $E(Y)$ are defined by the following iterative procedure:

\begin{enumerate}
\item Construct initial estimators $p_M^0$, $p_X^0$, and $\mu^0$ for $p_M$, $p_X$, and $\mu$, respectively. We use the same initial estimators as  in Section~\ref{sec_onestep}.
\item Construct
  a sequence of updated density estimates $p^k$, $k=1,2,\dots$, where at each iteration $k$ we construct $p^{k+1}$ as follows:
  estimate $\epsilon$ as
  \[\hat\epsilon=\arg\max_{\epsilon} \sum_{i=1}^n \log p_\epsilon^k(X_i,M_i,Y_i),\]
  where $p_\epsilon$ is given by (\ref{ff1}) or (\ref{ff2}), for the third or fourth implementation, respectively.
 Computation of $p_\epsilon^k(X_i,M_i,Y_i)$ requires evaluation of
  $D(p^k,O)$, which in turn requires $p_M^k$, $p_X^k$, and $\mu^k$.
  Define $p^{k+1}=p_{\hat\epsilon}^k$, and define  $p_M^{k+1}$,
  $p_X^{k+1}$, $\mu^{k+1}$ to be the corresponding components of $p^{k+1}$.
\item The previous step is iterated until convergence, i.e., until $\hat\epsilon\approx 0$. Denote the last step of the procedure by $k=k^*$.
\item The TMLE of $\psi_0$ is defined as the substitution estimator
  $\hat\psi\equiv \Psi(p^{k^*}) = E_{p_X^*}\{\mu^*(X)\}$, for
  $p_X^*$ and $\mu^*$ the corresponding components of $p^{k^*}$.
\end{enumerate}

If the initial estimator $p_X^0$ is the empirical distribution, then
$p_X^*$ is a density (with respect to counting measure) with positive mass only at the observed
values $X_i$. This is an important computational characteristic when computing the normalizing constant $c(\epsilon, p^*)$, since
integrals over $p_X^*$ become weighted sums over the sample. The
optimization in each iteration of step 2 is carried out using the BFGS
\cite{Broyden1970,Fletcher1970,Goldfarb1970,Shanno1970}
algorithm as implemented in the $R$ function {\texttt optim()}. The
optimization problem is convex in $\epsilon$, so that under regularity
conditions we expect the algorithm to converge to the global optimum.

\paragraph{Motivation for TMLE implementation with submodel
  (\ref{ff2}).} Submodel (\ref{ff1}) is not necessarily well define
for an unbounded efficient influence function $D(p^k,O)$. However, submodel
(\ref{ff2}) is always bounded and can be used with any
$D(p^k,O)$. An example of an unbounded efficient influence function is
given by (\ref{eicEY}) under empirical violations of the assumption
$P(p_M(X)>0)=1$. This problem has been extensively discussed,
particularly in the context of continuous outcomes
\cite[e.g.,][]{Bang05,Robins2007,Gruber2010t}. A TMLE with
submodel (\ref{ff2}) as presented in this section may provide an
alternative solution to those presented in the literature.


\subsection{Evaluating sensitivity of the TMLE to the choice of
  parametric submodel}\label{simula}
We perform a simulation study to explore the sensitivity of the TMLE to
the above four different choices of parametric submodels from
Sections~\ref{sec_onestep}-~\ref{ex1_exp_family_implementation}.
We generate data satisfying the missing at random assumption defined in Section~\ref{sec_prob_defn_1}.

\paragraph{Data generating mechanism for simulations}
The observed data on each  participant is the vector $(X,M,MY)$, where $X=(X_1,X_2)$.
The following defines the joint distribution of the variables $(X,Y)$:
\begin{align*}
  X_1 & \sim N(0,1/2),\\
  X_2|X_1 & \sim N(X_1,1),\\
  Y|X_1,X_2 & \sim Ber(\logit(X_2-X_2^2)),
\end{align*}
where $Ber(p)$ denotes the Bernoulli distribution with probability $p$ of $1$ and probability $1-p$ of $0$.
We consider the following three missing outcome distributions, which are referred to as missingness mechanisms, and are depicted in Figure~\ref{Missing}:
\begin{align}
  M|X_1,X_2 & \sim Ber(\logit(1+2X_2)),\label{D1}\\
  M|X_1,X_2 & \sim Ber(\logit(-1+2X_2)),\label{D2}\\
  M|X_1,X_2 & \sim Ber(\logit(-6+2X_2+2X_2^2)).\label{D3}
\end{align}

\begin{figure}[!htb]
  \caption{Different missingness mechanisms $p_M(X)$ considered.}
  \centering
  \includegraphics[scale = 0.4]{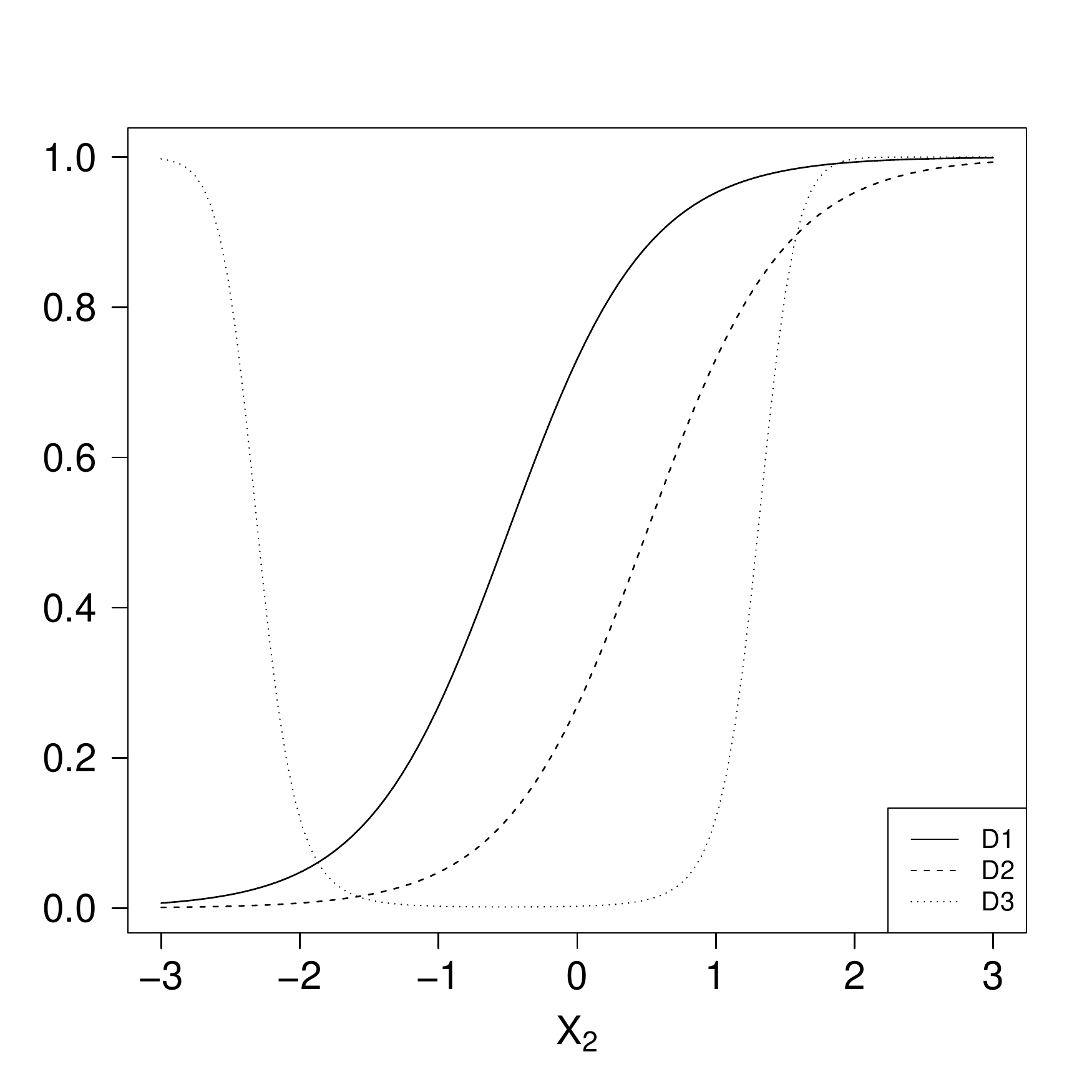}\label{Missing}
\end{figure}
We refer to these missingness mechanisms as D1, D2, and D3 respectively.
A practical violation of the positivity assumption is said to occur if
for some values of $(X_1,X_2)$, we have $P(M=1|X_1,X_2) \approx 0$.
Practical positivity violations are moderate under D2 and severe under D3.
We consider these three missingness mechanisms to assess the
performance of each TMLE implementation under different practical
violations to the positivity assumption, a scenario of high interest
since many doubly robust estimators can perform
poorly \cite{Robins2007}. A large fraction of small
probabilities as in D3 may be unlikely in a
missing data application. However, it is very common in survey sample
estimation, a field in which inverse probability weighted
estimators are the rule. For reference, the minimum missingness
probability in D3 is $0.0015$, and the median is $0.0047$. This is
consistent with survey weights found in the literature \cite[e.g.,][]{Rudolph14}.

Various studies have investigated the
performance of different estimators under violations to the positivity
assumption \cite[e.g.,][]{Kang&Shafer07, Porter2011}. We focus
on TMLEs, assessing the impact of the choice of submodel. Each
missingness mechanism, combined with the joint distribution of $(X,Y)$
defined above, determines the joint distribution of the observed data
$(X,M,MY)$.
We simulated 10000 samples of sizes 200, 500, 1000, and 10000,
respectively. This was done for each missingness mechanism.

We implemented the four types of TMLE  described
in this paper, using four different sets of working models for $\mu$ and
$p_M$: (i) correctly specified models for both, (ii)
correct model for $\mu$ and incorrect model for $p_M$, (iii) incorrect
model for $\mu$ and correct model for $p_M$, (iv) incorrect models
for both $\mu$ and $p_M$.
Misspecification of the working model for  $\mu$ consisted of using a logistic regression of $Y$ on $(X_1,X_1^2)$ among individuals with
$M=1$; misspecification of the working model for $p_M$ consisted of running
logistic regressions of $M$ on $(X_1, X_1^2)$. The
TMLE iteration was stopped whenever $\hat\epsilon < 10^{-4}$.

\paragraph{Simulation results}

Table \ref{msetable} shows the relative efficiency (using as reference
the analytically computed efficiency bound) of the four
estimators for different sample sizes under each working model
specification. The efficiency bounds for distributions D1, D2, and D3
are 0.34, 1.05, and 55.23, respectively.

The estimators with model specification (i) would be expected to have
asymptotic relative  efficiency equal to 1, which they approximately do at sample size
10000. The MSE of all estimators under severe positivity violations
(missingness mechanism D3) and  model specifications (i), (ii), and (iii) is
smaller than the efficiency bound for sample sizes 200 and 500. This fact does not contradict
theory; it is expected since the TMLE is a substitution
estimator and therefore has variance bounded between zero and one,
even when the efficiency bound divided by the sample size falls outside of this interval. A
similar observation is true for model specification (ii) for all sample
sizes. In this case, misspecification of the missingness model causes
a substantial reduction in variability of the inverse probability
weights, which results in smaller finite sample
variance. However, we note that the relative efficiency gets closer to
its theoretical value of one as the sample size increases. In the
extreme case of D3 a sample size of 10000 was not large enough to observe
the properties predicted by asymptotic theory.

It was not expected that under D1, the TMLE estimators are approximately semiparametric efficient
for models (ii) and (iii). This may be a  particularity of this
data generating mechanism, perhaps due to the low dimension of
the problem and the smoothness of this data generating
mechanism. As expected, the bias of the estimators
times square root of $n$ does not converge for model specification
(iv), since asymptotic theory dictates that in general at least one
of the working models must be correctly specified to imply
consistency of the TMLE.

Table~\ref{pcbiastable} shows the percent bias associated with each
estimator. The second TMLE implementation 2 performs better than its competitors in
finite samples under severe violations to the positivity assumption
(D3), particularly when the missingness model is
correctly specified (model specifications (i) and (iii)). However, under moderate positivity violation
(D2) and misspecification of the outcome  mechanism (model specification (iii)) in large samples ($n=10000$), TMLE
implementation 2 performed worse than all of its competitors (MSE 1.63 vs
1.06). Implementations 3 and 4 did not perform particularly better
than the best of implementations 1 and 2 for any scenario under consideration.

\begin{table}[!htb]
  \caption{Relative performance of the estimators ($n$ times MSE divided by the
    efficiency bound). $\hat\psi_1$, $\hat\psi_2$, $\hat\psi_3$ and
    $\hat\psi_4$ correspond to TMLE implementations 1, 2, 3, and 4,
    respectively. Data generating mechanisms D1, D2, and D3 correspond to
    expressions (\ref{D1}), (\ref{D2}), and (\ref{D3}), respectively. Model specification
    (i) is correctly specified for $\mu$ and $p_M$, (ii) is correct
    for $\mu$ and incorrect for $p_M$. (iii) is incorrect for $\mu$ and
    correct for $p_M$, and (iv) is incorrect for both.}
  \centering\scriptsize

  \begin{tabular}{rc|cccc|cccc|cccc}\hline
    \multirow{3}{*}{$n$} &               & \multicolumn{12}{c}{Missingness Mechanism}                                \\
    & Working              & \multicolumn{4}{c}{D1} & \multicolumn{4}{c}{D2} & \multicolumn{4}{c}{D3} \\
    &  model &
    $\hat\psi_1$ & $\hat\psi_2$ & $\hat\psi_3$
    & $\hat\psi_4$ & $\hat\psi_1$ &
    $\hat\psi_2$ & $\hat\psi_3$ & $\hat\psi_4$ &$\hat\psi_1$ & $\hat\psi_2$ & $\hat\psi_3$ & $\hat\psi_4$                                \\\hline
    \multirow{4}{*}{200}   & (i)           & 1.09  & 1.03  & 1.07  & 1.05  & 1.49  & 1.21  & 1.36  & 1.38  & 0.48  & 0.34  & 0.39  & 0.42  \\
    & (ii)          & 1.11  & 1.14  & 1.11  & 1.11  & 1.28  & 1.29  & 1.27  & 1.25  & 0.25  & 0.25  & 0.26  & 0.27  \\
    & (iii)         & 1.07  & 1.17  & 1.10  & 1.08  & 1.82  & 1.56  & 1.73  & 1.86  & 0.58  & 0.35  & 0.44  & 0.47  \\
    & (iv)          & 2.98  & 2.98  & 2.97  & 2.97  & 2.85  & 2.84  & 2.24  & 2.23  & 1.17  & 1.17  & 1.18  & 1.18  \\ \hline
    \multirow{4}{*}{500}   & (i)           & 1.02  & 1.00  & 1.01  & 1.00  & 1.26  & 1.03  & 1.17  & 1.12  & 0.83  & 0.58  & 0.69  & 0.70  \\
    & (ii)          & 1.03  & 1.04  & 1.03  & 1.03  & 1.12  & 1.17  & 1.12  & 1.11  & 0.31  & 0.31  & 0.31  & 0.32  \\
    & (iii)         & 1.07  & 1.18  & 1.09  & 1.08  & 1.43  & 1.54  & 1.48  & 1.45  & 0.94  & 0.59  & 0.77  & 0.76  \\
    & (iv)          & 5.50  & 5.51  & 5.50  & 5.50  & 5.34  & 5.34  & 3.78  & 3.77  & 2.74  & 2.74  & 2.75  & 2.75  \\ \hline
    \multirow{4}{*}{1000}  & (i)           & 1.01  & 1.00  & 1.01  & 1.00  & 1.19  & 1.05  & 1.14  & 1.07  & 1.09  & 0.79  & 0.94  & 0.93  \\
    & (ii)          & 1.06  & 1.07  & 1.06  & 1.06  & 1.11  & 1.16  & 1.10  & 1.10  & 0.37  & 0.37  & 0.37  & 0.38  \\
    & (iii)         & 1.01  & 1.11  & 1.02  & 1.02  & 1.18  & 1.56  & 1.29  & 1.18  & 1.25  & 0.86  & 1.10  & 1.07  \\
    & (iv)          & 9.89  & 9.89  & 9.89  & 9.89  & 9.74  & 9.74  & 6.29  & 6.29  & 5.27  & 5.26  & 5.27  & 5.27  \\\hline
    \multirow{4}{*}{10000} & (i)           & 0.99  & 0.99  & 0.99  & 0.99  & 1.04  & 0.99  & 1.01  & 0.99  & 1.07  & 1.00  & 1.04  & 1.07  \\
    & (ii)          & 1.05  & 1.06  & 1.06  & 1.06  & 1.06  & 1.11  & 1.06  & 1.06  & 0.54  & 0.53  & 0.54  & 0.54  \\
    & (iii)         & 1.01  & 1.10  & 1.01  & 1.01  & 1.06  & 1.63  & 1.03  & 1.03  & 1.36  & 1.31  & 1.53  & 1.24  \\
    & (iv)          & 87.40 & 87.40 & 87.40 & 87.40 & 87.26 & 87.25 & 53.14 & 53.14 & 52.11 & 52.11 & 52.10 & 52.10 \\\hline
  \end{tabular}
  \label{msetable}
\end{table}

\begin{table}[!htb]
  \caption{Percent bias of each estimator. The true value of the
    parameter is 0.36. $\hat\psi_1$, $\hat\psi_2$, $\hat\psi_3$ and
    $\hat\psi_4$ correspond to TMLE implementations 1, 2, 3, and 4,
    respectively. Data generating mechanisms D1, D2, and D3 correspond to
    expressions (\ref{D1}), (\ref{D2}), and (\ref{D3}), respectively. Model specification
    (i) is correctly specified for $\mu$ and $p_M$, (ii) is correct
    for $\mu$ and incorrect for $p_M$. (iii) is incorrect for $\mu$ and
    correct for $p_M$, and (iv) is incorrect for both.}
  \centering\scriptsize

  \begin{tabular}{rc|cccc|cccc|cccc}\hline
    \multirow{3}{*}{$n$} &               & \multicolumn{12}{c}{Missingness Mechanism}                                \\
    & Working              & \multicolumn{4}{c}{D1} & \multicolumn{4}{c}{D2} & \multicolumn{4}{c}{D3} \\
    &  model &
    $\hat\psi_1$ & $\hat\psi_2$ & $\hat\psi_3$
    & $\hat\psi_4$ & $\hat\psi_1$ &
    $\hat\psi_2$ & $\hat\psi_3$ & $\hat\psi_4$ &$\hat\psi_1$ & $\hat\psi_2$ & $\hat\psi_3$ & $\hat\psi_4$                                \\\hline
    \multirow{4}{*}{200}   & (i)   & 0.40  & 0.00  & 0.20  & 0.10  & 3.40  & 0.80  & 1.90  & 1.70  & 25.00 & 5.50  & 3.00  & 9.50   \\
    & (ii)  & 0.20  & 0.20  & 0.10  & 0.10  & 1.50  & 0.90  & 1.30  & 0.80  & 18.20 & 17.90 & 17.90 & 17.30  \\
    & (iii) & 0.10  & 1.30  & 0.60  & 0.20  & 5.30  & 6.80  & 6.70  & 6.00  & 39.30 & 8.00  & 11.30 & 22.50  \\
    & (iv)  & 15.20 & 15.20 & 15.20 & 15.20 & 21.50 & 21.40 & 21.50 & 21.50 & 33.10 & 33.90 & 33.20 & 33.10  \\ \hline
    \multirow{4}{*}{500}   & (i)   & 0.10  & 0.00  & 0.10  & 0.00  & 1.20  & 0.00  & 0.60  & 0.20  & 20.10 & 2.90  & 8.40  & 10.80  \\
    & (ii)  & 0.10  & 0.10  & 0.10  & 0.10  & 0.80  & 0.60  & 0.70  & 0.60  & 17.50 & 17.50 & 17.40 & 17.30  \\
    & (iii) & 0.10  & 0.60  & 0.40  & 0.10  & 1.40  & 3.20  & 2.60  & 1.70  & 31.00 & 2.70  & 20.90 & 27.40  \\
    & (iv)  & 15.10 & 15.10 & 15.10 & 15.10 & 21.10 & 21.10 & 21.10 & 21.10 & 33.30 & 33.80 & 33.30 & 33.30  \\ \hline
    \multirow{4}{*}{1000}  & (i)   & 0.00  & 0.10  & 0.10  & 0.10  & 0.70  & 0.00  & 0.30  & 0.10  & 8.80  & 3.40  & 4.00  & 5.60   \\
    & (ii)  & 0.00  & 0.00  & 0.00  & 0.00  & 0.30  & 0.30  & 0.30  & 0.20  & 14.40 & 14.50 & 14.30 & 14.30  \\
    & (iii) & 0.00  & 0.30  & 0.30  & 0.10  & 0.30  & 1.70  & 1.30  & 0.50  & 21.50 & 0.40  & 17.70 & 21.70  \\
    & (iv)  & 15.30 & 15.30 & 15.30 & 15.30 & 20.80 & 20.80 & 20.80 & 20.80 & 32.80 & 33.30 & 32.90 & 32.90  \\\hline
    \multirow{4}{*}{10000} & (i)   & 0.00  & 0.00  & 0.00  & 0.00  & 0.10  & 0.00  & 0.00  & 0.00  & 0.10  & 0.30  & 0.20  & 0.10   \\
    & (ii)  & 0.00  & 0.00  & 0.00  & 0.00  & 0.10  & 0.10  & 0.00  & 0.00  & 3.40  & 3.40  & 3.10  & 3.10   \\
    & (iii) & 0.00  & 0.00  & 0.20  & 0.00  & 0.00  & 0.20  & 0.40  & 0.10  & 2.40  & 0.20  & 4.70  & 0.70   \\
    & (iv)  & 15.20 & 15.20 & 15.20 & 15.20 & 20.80 & 20.80 & 20.80 & 20.80 & 32.30 & 32.80 & 32.40 & 32.30  \\\hline
  \end{tabular}
  \label{pcbiastable}
\end{table}
\newpage
Another important question to ask when deciding on a parametric
submodel is the computational efficiency of the estimators.
Our simulations are in accordance to what we have observed in practice for this and other parameters, in  that TMLE typically requires 6 or fewer iterations.
In the above simulations, the  time required to compute the TMLE for a single data set
was typically less than a second.

\section{Example 2: Median regression}\label{ex2}
Consider the median regression model:
\begin{equation}
  Y = g(X,\beta) + \delta \label{med_regr_model1},
\end{equation}
where $g(X,\beta)$ is a known, smooth function in $\beta$,
and where the conditional median of $\delta$ given $X$ is $0$ a.s.
This last condition is equivalent to having with probability 1 that
\begin{equation}
  P(\delta \geq 0|X) \geq1/2 \mbox{ and }  P(\delta \leq 0|X) \geq 1/2.  \label{med_regr_model2}
\end{equation}
We let $\lambda(x,y)$ denote a dominating measure for the
distributions $P$ we consider, and denote by $p$ the density of
$P$. We say the above median regression model is correctly specified
if at the true data generating distribution $\Pt$, we have for some
$\beta_0$ that the conditional median under $\Pt$ of $Y - g(X,\beta_0)$ given $X$ is
$0$, with probability 1. Throughout, we do not assume the median regression model is correctly specified.

Define the following nonparametric extension of $\beta$ (which maps each density $p$ to a value
$\beta^\ast(p)$ in $\mathbb{R}^d$):
\begin{equation}
  \beta^\ast(p) \equiv \arg \min_\beta E_p |Y -
  g(X,\beta)|. \label{parameter_defn}
\end{equation}
We assume there is a unique minimizer in $\beta$
of $E_p |Y - g(X,\beta)|$. Under this assumption, if the median regression model
(\ref{med_regr_model1},\ref{med_regr_model2}) is correctly specified,
then this unique minimizer equals $\beta^\ast(p)$. However, even
when (\ref{med_regr_model1},\ref{med_regr_model2}) is misspecified,
the parameter in (\ref{parameter_defn}) is well defined as long as
there is a unique minimizer of $E_p |Y - g(X,\beta)|$. To simplify the notation, we denote $\beta^\ast(p)$ by $\beta(p)$ and
$\beta(\pt)$ by $\beta_0$. The goal is to estimate $\beta_0$ based on $n$ i.i.d. draws
$O_i = (X_i,Y_i)$ from an unknown data generating distribution
$\Pt$.

The nonparametric estimator of $\beta_0$ is the minimizer in $\beta$ of $\frac{1}{n}\sum_{i=1}^n
|Y_i - g(X_i,\beta)|$. \citet{Koenker1996} proposed a solution to this
optimization problem based on linear programming. Their methods are
implemented in the \texttt{quantreg} \cite{quantreg} R package. We develop a TMLE for $\beta_0$, in order to demonstrate it is a general methodology that can be applied to a variety of estimation problems, and to compare its performance versus the estimator of \citet{Koenker1996} that is explicitly tailored to the problem in this section.

The efficient influence function for the parameter
(\ref{parameter_defn}) in the nonparametric model, at distribution $P$, is
(up to a normalizing constant which we suppress in what follows):
\begin{eqnarray}
  D(p,X,Y) \equiv -\left.\frac{d}{d\beta}g(X,\beta)\right|_{\beta =
    \beta(p)}\mbox{sign}\{Y-g(X,\beta(p))\}. \label{eif}
\end{eqnarray}
In particular, we have
\begin{equation}
  E_pD(p,X,Y) = 0, \label{score_property}
\end{equation}
for all sufficiently smooth $p$. 

\paragraph{Construction of parametric  submodel}
Given the estimate $p^k$ of $p_0$ at iteration $k$ of the TMLE algorithm, we construct
a regular, parametric model $\{p_\epsilon^k:\epsilon\}$
satisfying: (i) $p_0^k=p^k$ and (ii) $\frac{d}{d\epsilon}\log p_\epsilon^k(x,y)|_{\epsilon=0}=D(p^k,x,y)$ for each $x \in \mathcal{X}, y\in\mathcal{Y}$. We again use an exponential submodel as in (\ref{expfam1}), which in this case is
\begin{equation}
  p_\epsilon^k(x,y) = p^k(x,y)\exp(\epsilon D(p^k,x,y))c(\epsilon,p^k), \label{submodel}
\end{equation}
where the normalization constant
$c(\epsilon,p) = \left[\int p(x,y)\exp(\epsilon D(p,x,y))d\lambda(x,y)\right]^{-1}$.
This parametric model is well-defined, regular, equals $p^k$ at $\epsilon=0$, and has score:
\begin{equation}
  \frac{d}{d\epsilon}\left[\log p_\epsilon^k (x,y)\right] = D(p^k,x,y) -
  c(\epsilon, p^k)\int D(p^k,x,y)\exp\{\epsilon D(p^k,x,y)\}p^k(x,y)d\lambda(x,y), \label{score1}
\end{equation}
under smoothness and integrability conditions that allow the interchange of the order of differentiation and
integration. It follows from
 (\ref{score_property}) and (\ref{score1}) that the
score equals $D(p^k,x,y)$ at $\epsilon=0$, and therefore satisfies the conditions
(i) and (ii) described above.

\paragraph{Implementation of Targeted Maximum Likelihood Estimator}
We present an implementation of the TMLE applying the parametric submodel (\ref{submodel}). First, we
construct an initial density estimator $p^0(x,y)$. We let $p^0(x)$ be
the empirical distribution of $X$. We
fit a linear regression model for $Y$
given $X$ with main terms only, and let $p^0(y|X=x)$ be a normal
distribution with conditional mean as given in the linear regression
fit, and with conditional variance 1. We then define $p^0(x,y) \equiv
p^0(y|x)p^0(x).$ A more flexible method can be
used to construct the initial fit for the density of $Y$ given
$X$. Here we use this simple model to examine how well the
TMLE can recover from a poor choice for the initial density estimate.

Initializing $k=0$, the iterative procedure defining the TMLE involves the following
computations, at each iteration $k$ (where $a \leftarrow b$ represents setting $a$ to take value $b$):
\begin{eqnarray}
  \beta^k & \leftarrow &\arg \min_\beta E_{p^k} |Y - g(X,\beta)|; \label{alg_step1} \\
  D(p^k,X,Y) & \leftarrow & -\left.\frac{d}{d\beta}g(X, \beta)\right|_{\beta = \beta(p^k)}\mbox{sign}(Y-g(X,\beta(p^k))); \label{alg_step2} \\
  p_\epsilon^k(x,y) & \leftarrow & p^k(x,y)\exp(\epsilon D(p^k,x,y))c(\epsilon,p^k); \label{submodel_again} \\
  \hat\epsilon & \leftarrow &\arg \max_\epsilon \sum_{i=1}^n \log
  p^k(X_i,Y_i)\exp\{\epsilon
  D(p^k,X_i,Y_i)\}c(\epsilon,p^k); \label{alg_step3} \\
  p^{k+1} & \leftarrow & p^k_{\hat\epsilon}.
\end{eqnarray}
The value of $\beta^k$ in (\ref{alg_step1}) is approximated by
grid search over $\beta$, where for each value of $\beta$ considered,
the expectation on the right hand side of (\ref{alg_step1}) is
approximated by Monte Carlo integration (based on generating 10000
independent realizations from the density $p^k$ and taking the empirical
mean over these realizations). The value of $\hat\epsilon$ in
(\ref{alg_step3}) is approximated by applying the Newton-Raphson
algorithm to the summation on the right side of
(\ref{alg_step3}), where we use the analytically derived gradient and Hessian
in the Newton-Raphson algorithm. The above process is
iterated over $k$ until convergence (we used $\hat\epsilon < 10^{-4}$
as stopping rule).

The density at the final iteration is denoted by
$p^\ast$, and the TMLE of $\beta_0$ is defined as $\beta_n\equiv
\beta(p^\ast)$. If the initial estimator of
$\pt(Y|X)$ is consistent, it is possible to use standard arguments for the analysis of
targeted maximum likelihood estimators \cite{vanderLaanRose11} to show
that this estimator is asymptotically linear with influence function
equal to the efficient influence function $D(\pt,X,Y)$.

\paragraph{Simulation} We draw 10000 samples, each of
size 1000, of a two dimensional covariate $X=(X_1,X_2)$ by drawing $X_1\sim
U(0,1)$ and $X_2\sim U(0,1)$, with $X_1,X_2$ independent. We consider
two different outcome distributions for $Y$ given $X$; the outcomes
under each distribution are denoted by $Y_1,Y_2$, respectively.
The first outcome involves drawing $\delta_1\sim Exp(3)$, and setting
\begin{equation}
  Y_1=-\frac{\ln(2)}{3} +
  \expit(1.5X_1 + 2.5X_2) + \delta_1,\label{DR1}
\end{equation} where $\expit$ is the inverse of the
$\logit$ function $\logit(x)=\log(x/(1-x))$. This represents a case in which the error
distribution is skewed. The constant $\ln(2)/3$ was selected since it is the median of $Exp(3)$, which implies the median $Y_1$ given $X$  equals
 $\expit(1.5X_1 + 2.5X_2)$.
The second outcome distribution involves drawing
$\delta_2\sim N(0,1)$ and setting
\begin{equation}
  Y_2=\exp(X_1 + 2X_2) + \delta_2.\label{DR2}
\end{equation}
Note that we used $\exp$ in the above display instead of $\expit$.
We denote the first outcome distribution
    (\ref{DR1}) by $D1$, and the second (\ref{DR2}) by $D2$.

We are interested in estimating the parameters
\begin{eqnarray}
  \beta_0^{(1)} &=& \arg\min_\beta E|Y_1 - g(X,\beta)| \label{Y_1_param} \\
  \beta_0^{(2)} &=&  \arg\min_\beta E|Y_2 - g(X,\beta)|, \label{Y_2_param}
\end{eqnarray}
where $g(X,\beta)=\expit(\beta' X)$. The true value of $\beta_0^{(1)}$
is $(1.5, 2.5)$, whereas the true value of $\beta_0^{(2)}$ is
approximately $(2.1, 9.2)$ (obtained using Monte Carlo simulation).
The median regression model $g(X,\beta)=\expit(\beta' X)$
 is a correctly specified model under the first outcome distribution, but is
incorrectly specified for the second outcome distribution.
However, the minimizers $\beta_0^{(1)}$ and
$\beta_0^{(2)}$ of the right sides of (\ref{Y_1_param}) and (\ref{Y_2_param}), respectively, are both well defined.

For each of the 10000 samples we computed the estimator described
above. The marginal distribution of $X$ was estimated by the empirical distribution in the given sample. The conditional distribution of $Y$ given
$X$ was misspecified by running a linear regression of $Y$ on
$(X_1,X_2)$ with main terms only,
and assuming that $Y$ is normally distributed with conditional
variance equal to one. This was done in order to assess how the TMLE
can recover from a poor fit of the initial densities resulting from a
distribution that is commonly used in statistical practice. We then computed
the MSE across the 10000 estimates as $E(||\hat \beta - \beta_0||^2)$,
where $||\cdot||$ represents the Euclidean norm. The results are
presented in Table \ref{tabMedReg}. For comparison, we
computed the same results for two other estimators: the quantile
regression function \texttt{nlrq()} implemented in the R package
\texttt{quantreg}, and a substitution estimator (SE) that is the
result of optimizing (\ref{alg_step1}) in the first iteration with
$p^k$ set to $p^0$.

\begin{table}[ht]
  \caption{Square root of mean squared error (MSE) for estimators of
    parameters (\ref{Y_1_param}) and (\ref{Y_2_param}). }\label{tabMedReg}
  \centering
  \begin{tabular}{rrr|rrr}
    \hline
     \multicolumn{3}{c|}{$\beta_0^{(1)}$}          & \multicolumn{3}{c}{$\beta_0^{(2)}$}\\\hline
     TMLE & \texttt{nlrq()} & SE    & TMLE & \texttt{nlrq()} & SE\\
    \hline
     0.37 & 0.38 & 3.99 & 7.15 & 8.50 & 6.76 \\
    \hline
  \end{tabular}
\end{table}

The TMLE and the estimator of \citet{Koenker1996} perform
similarly for $\beta_0^{(1)}$, i.e., when the median regression model is correct.
The TMLE and SE perform better for estimating $\beta_0^{(2)}$, i.e., when
the median regression model is incorrectly specified. This is not
surprising since the estimator of \citet{Koenker1996} is designed for
the case where the median regression model is correctly specified.

\section{Example 3: The causal effect of a continuous exposure} \label{ex3}

We explore the use of an exponential family as a
parametric submodel only for certain components of the likelihood. Consider
a continuous exposure $A$, a binary outcome $Y$, and a set of covariates
$W$. For a user-given value $\gamma$ we are interested in estimating
the expectation of $Y$ under an intervention that causes a shift of
$\gamma$ units in the distribution of $A$ conditional on $W$. Formally, consider an i.i.d. sample of $n$ draws of the random variable $O=(W,A,Y)\sim \Pt$. Denote
$\mu_0(A,W)\equiv E_{\pt}(Y|A,W)$, $p_{W,0}(W)$  the marginal
density of $W$ and $p_{A,0}(A|W)$  the conditional density of $A$
given $W$. We assume that these data were generated by a nonparametric
structural equation model \cite[NPSEM,][]{Pearl00}:
\[  W=f_W(U_W);\quad
A=f_A(W, U_A);\quad
Y=f_Y(A,W,U_Y),
\]
where $f_W$, $f_A$, and $f_Y$ are unknown but fixed functions, and
$U_W$, $U_A$, and $U_Y$ are exogenous random variables satisfying the
randomization assumption $U_A\indep U_Y|W$. We are interested in the
causal effect  on $Y$ of a shift of $\gamma$ units in $A$. Consider the following intervened NPSEM
\[  W=f_W(U_W);\quad
A_\gamma=f_A(W, U_A)+\gamma;\quad
Y_\gamma=f_Y(A_\gamma,W,U_Y).
\]
This intervened NPSEM represents the random variables that would have
been observed in a hypothetical world in which every participant
received $\gamma$ additional units of exposure $A$. \citet{Diaz12} proved that
\[E(Y_\gamma)=E_0\{\mu_0(A+\gamma, W)\}.\]
For each density $p$, define the parameter
\[\Psi(p)\equiv E_{p_W, p_A}\{\mu(A+\gamma, W)\},\]
where $\mu$, $p_A$, and $p_W$ are the outcome conditional
expectation, exposure mechanism, and
covariate marginal density corresponding to $p$,
respectively. We also use the notation $\Psi(\mu, p_A, p_W)$ to
refer to $\Psi(p)$. We are interested in estimating the true value of the parameter
$\psi_0\equiv \Psi(\pt)$.

The efficient influence function of $\Psi(p)$ at $p$ is given by \cite{Diaz12} as:
\[D(p,O)\equiv \frac{p_A(A-\gamma|W)}{p_A(A|W)}\{Y-\mu(A, W)\} +
\mu(A+\gamma, W) - \Psi(p).\]

\paragraph{Construction of the parametric submodel}
Consider initial estimators $\mu^0(A, W)$ and $p_A^0(A|W)$, which can be
obtained, for example, through machine learning methods. We estimate
the marginal density of $W$ with its empirical
counterpart denoted $p_W^0$, and construct a sequence of parametric submodels for
$p_0$ by specifying each component as:
\begin{align*}
  \logit \mu_\epsilon^k(A,W)&= \logit \mu^k (A, W) + \epsilon
  H_Y^k(A, W)\\
  p_{A,\epsilon}^k(A|W) & = c_1(\epsilon, W)p_A^k(A|W)\exp\{\epsilon H_A^k(A,W)\}\\
  p_{W, \theta}^k(W) & = c_2(\theta, W)p_W^k(W)\exp\{\theta H_W^k(W)\},
\end{align*}
where
\begin{align*}
  H_Y^k(A,W)&=\frac{p_A^k(A-\gamma|W)}{p_A^k(A|W)}\\
  H_A^k(A,W)&=\mu^k(A+\gamma, W) - E_{p^k}\{\mu^k(A+\gamma, W)|W\}\\
  H_W^k(W) & = E_{p^k}\{\mu^k(A+\gamma, W)|W\} - \Psi(p^k),
\end{align*}
and $c_1$, $c_2$ are the corresponding normalizing constants. The sum of the scores of these models at $\epsilon=0,\theta=0$ equals the
efficient influence function $D(p^k,O)$.

\paragraph{Implementation of Targeted Maximum Likelihood Estimator}
Following the TMLE template of Section \ref{tmle}, we have (where
the MLE of $\theta$ is $0$, due to the initial estimator of $p_W$ being the empirical distribution):
\begin{enumerate}
\item Compute initial estimators $\mu^0(A, W)$ and $p_A^0(A|W)$. For
  example, $\mu^0$ and $p_A^0$ may be obtained through a stacked
  predictor such as super learning \cite{vanderLaan&Polley&Hubbard07}. Super learners rely on a
  user-supplied library of prediction algorithms to build a convex
  combination with weights that minimize the cross-validated empirical
  risk. In particular, the authors in the original paper \cite{Diaz12}
  use a library containing various specifications of generalized linear models (GLMs), generalized
  additive models, and Bayesian GLMs for $\mu^0$. The conditional
  density $p_A^0$ is estimated using super learning in a library of
  histogram-like density estimators, as proposed in
  \cite{Diaz11}. Here we emphasize a particularly appealing feature of TMLE:
  it allows the integration of machine learning methods with
  semiparametric efficient estimation.
\item Construct
  a sequence of updated density estimates $p^k$, $k=1,2,\dots$, where at each iteration $k$ we
  estimate the maximizer of the relevant parts of
  the log likelihood: \[\hat\epsilon =
  \arg\max_{\epsilon}\sum_{i=1}^n\{Y_i\log\mu^k_\epsilon(A_i,W_i) +
  (1-Y_i)\log(1-\mu^k_\epsilon(A_i,W_i)) + \log p_{A,\epsilon}^k(A_i,W_i)\}.\]
  Define $p^{k+1}=p^{k}(\hat\epsilon)$
\item The previous step is iterated until convergence, i.e., until $\hat\epsilon\approx 0$. Denote the last step of the procedure by $k=k^*$.
\item The TMLE of $\psi_0$ is defined as the substitution estimator
  $\hat\psi\equiv \Psi(p^{k^*})$.
\end{enumerate}
Optimization of the likelihood in step 2 is a convex optimization
problem that may be solved, for example, based on the Newton-Raphson
algorithm. Another option, implemented in the original paper using the R function \texttt{uniroot()}, is to
solve the estimating equation $S^k(\epsilon)=0$, where
\begin{multline}
  S^k(\epsilon) = \sum_{i=1}^n\bigg\{[Y_i - \mu^k(A_i,W_i) -
  \epsilon_1H_Y^k(O_i)]H_Y^k(O_i) + H_A^k(O_i) -
  \\\frac{\int H_A^k(a, W_i)\
    \exp\{\epsilon H_A^k(a, W_i)\}\
    p_A^k(a|W_i)da}{\int \exp\{\epsilon
    H_A^k(a, W_i)\}\ p_A^k(a|W_i)da}\bigg\}.
\end{multline}
In practice, the iteration process is carried out until
convergence in the values of $\hat\epsilon$ to approximately $0$ is
achieved. We denote
$\mu^\ast$ and $g_n^\ast$ the last values of the iteration, and
define the TMLE of $\psi_0$  as
$\psi_{n}\equiv \Psi(\mu^\ast, p_A^\ast, p_W^0)$. The variance
of $\psi_{n}$ can be estimated by the empirical variance of
$D(\mu^*,p_A^*, p_W^0,O)$. This is a consistent estimator of the
variance if both $p_A^0$ and $\mu^0$ are consistent. Simulations
studying the properties of this estimator were performed in the
original paper. The results of those simulations confirm the double
robustness of the TMLE (robustness to misspecification of one of the
estimates $\mu^0$ or $p_A^0$), it asymptotic efficiency, and its
superiority when compared to the inverse probability weighted
estimator. For more extensive discussion of the properties of this estimator we refer the
reader to \cite{Diaz12}.

\section{Discussion}

We presented several implementations of TMLE using parametric families
with convex log-likelihood in three examples. Since reliable and efficient algorithms exist for
convex optimization, parametric submodels with
convex log-likelihood may lead to computationally advantageous implementations of TMLE.

An important choice in any TMLE implementation  is which
parametric submodel to use. We showed a simulation in which this choice
has a substantial impact on the performance of the targeted
maximum likelihood estimator, even at very large sample sizes. It is  important to
consider the characteristics of each submodel before choosing an
estimator.

An additional consideration for the choice
of parametric family is ease of implementation. For example, for estimation of the mean
of a binary outcome missing at random, the TMLE using a logistic parametric submodel
from Section \ref{sec_onestep} converges in one step and is generally
faster and easier to implement than the alternatives we considered.  However, under
severe violations of the positivity assumption, the
implementation of a TMLE presented in Section~\ref{second_onestep} may be more appealing than
the alternatives we considered, based on the results of our simulation.


Various estimators with desirable properties have been
proposed for some of the examples  in this paper. Notably,
estimation of the expectation of an outcome missing at random has been
widely studied (see, e.g., \citet{Kang&Shafer07} for a review). Also, \citet{Haneuse2013} propose an estimator
of $\psi_0$ in Example 3 that relies on the correct specification of a
parametric model for the outcome expectation and
the conditional density of the exposure. TMLE, on the other hand, is
a general estimation template that allows the construction of
estimators for a considerable class of statistical problems, allowing
integration with data-adaptive estimation methods.

\section*{Acknowledgements}
We would like to thank Mark van der Laan for helpful discussions
that greatly improved this manuscript.

\appendix

\section{Asymptotic equivalence of targeted maximum
  likelihood estimators using different submodels}\label{theorem}
\begin{theorem}\label{Theo} Let $O_1,\ldots,O_n$ be an i.i.d. sample
  from $O=(X,M,MY)\sim\Pt$. Assume $\Pt\in \mathcal{M}$, where
  $\mathcal{M}$ denotes the nonparametric model. Assume there exists
  a dominating measure $\nu$ for $\mathcal M$ so that the density
  $p$ of $P$ is well defined. Denote $\mu(X) \equiv E_{P}(Y|M=1,X)$, and
  $p_{M}(X)\equiv P(M=1|X)$. Define $\Psi(\mu,p_X)=E_{P_X}(\mu(X))$. Let
  $\mu^0$ and $p_M^0$ be estimators of $\mu_0$ and $p_{M,0}$. Let
  $p_X^0$ denote the empirical distribution of $X$.
  For fixed $p$, let $\mathcal M_1 =\{p_\epsilon:\epsilon\}$ and $\mathcal M_2
  =\{p_\theta:\theta\}$ be parametric submodels through $p$ (i.e.,
  $p_\epsilon|_{\epsilon=0}=p_\theta|_{\theta=0}=p$) satisfying
  \[\frac{d}{d\epsilon}\{\log
  p_\epsilon(O)\}\bigg|_{\epsilon=0}=\frac{d}{d\theta}\{\log
  p_\theta(O)\}\bigg|_{\theta=0}=D(p,O),\]
  where $D$ is defined in (\ref{eicEY}). Let $\hat \psi_1$ and $\hat \psi_2$ be the targeted maximum likelihood
  estimators of $\psi_0=\Psi(\mu_0,p_{X,0})$ using submodels $\mathcal
  M_1$ and $\mathcal M_2$, respectively.
  Assume:
  \begin{align*}
    &  \sqrt{n}\int (\mu^0(x) - \mu_0(x))^2dP_{X,0}(x)\to 0\\
    &  \sqrt{n}\int (p_M^0(x) - p_{M,0}(x))^2dP_{X,0}(x)\to 0
  \end{align*}
  as $n\to\infty$. In addition, assume that

\[
    \int (D(\mu^0, p_M^0, p_X^0, o) - D(\mu_0, p_{M,0}, p_{X,0}, o))^2dP(o)\to 0
\]
and that $D(\mu^0,p_M^0,p_X^0,O)$ belongs to a Donsker  class \cite{vanderVaart98}
with probability tending to one. Then we have
  \begin{align*}
    \sqrt{n}(\hat\psi_1-\psi_0)&\to N(0, \sigma^2)\\
    \sqrt{n}(\hat\psi_2-\psi_0)&\to N(0, \sigma^2),
  \end{align*}
  where $\sigma^2=Var(D(p_0,O))$.
\end{theorem}
\begin{proof}
  This results follows from Theorem 1 of \cite{vanderLaan&Rubin06}.
\end{proof}

\bibliographystyle{plainnat}
\bibliography{TMLE}

\end{document}